\documentclass[10pt,a4paper,onecolumn]{IEEEtran}
\usepackage{amsmath,amsfonts,amssymb,amsthm,graphicx,float,subfigure}
\def\EE{\mathbb E}
\def\Tr{\mathrm{Tr}}
\def\SNR{\mathrm{SNR}}
\newtheorem{stat}{Statement}[section]
\newtheorem{thm}[stat]{Theorem}
\newtheorem{lem}[stat]{Lemma}
\newtheorem{prop}[stat]{Proposition}

\usepackage{epstopdf}

\begin{document}

\title{On the Broadcast Capacity\\of Large Wireless Networks at Low SNR}
\author{Serj Haddad, Olivier L\'ev\^eque\\ Information Theory Laboratory -- School of Computer and Communication Sciences\\ EPFL, 1015 Lausanne, Switzerland
-- \{serj.haddad,olivier.leveque\}@epfl.ch \footnote{This paper was presented at ISIT 2015.}}
\maketitle
\thispagestyle{plain}
\pagestyle{plain}

\begin{abstract}
The present paper focuses on the problem of broadcasting information in the most efficient manner in a large two-dimensional {\it ad hoc} wireless network at low SNR and under line-of-sight propagation. A new communication scheme is proposed, where source nodes first broadcast their data to the entire network, despite the lack of sufficient available power. The signal's power is then reinforced via successive back-and-forth beamforming transmissions between different groups of nodes in the network, so that all nodes are able to decode the transmitted information at the end. This scheme is shown to achieve asymptotically the broadcast capacity of the network, which is expressed in terms of the largest singular value of the matrix of fading coefficients between the nodes in the network. A detailed mathematical analysis is then presented to evaluate the asymptotic behavior of this largest singular value.
\end{abstract}

\begin{IEEEkeywords}
wireless networks, broadcast capacity, low SNR communications, beamforming strategies, random matrices 
\end{IEEEkeywords}

%%%%%%%%%%%%%%%%%%%%%%%%
%%%%%%%%%%%%%%%%%%%%%%%%
%%%%%%%%%%%%%%%%%%%%%%%%
\section{Introduction}
The literature on the study of scaling laws in large {\it ad hoc} wireless networks concentrates mainly on \textit{multiple-unicast} (one-to-one) transmissions (see e.g.~\cite{GuptaUnicast,AyferUnicast}). This does not degrade by any means the importance of investigating \textit{multicast} (one-to-many) transmissions for several reasons such as the need of many network protocols to broadcast control signals or to enhance cooperation among nodes belonging to the same cluster or cell. In the present paper, we are interested in studying how can source nodes broadcast their data to the whole network in the most efficient way. Previous works investigated the broadcast capacity of wireless networks under specific channel models and mainly at high SNR \cite{GastparBC,KeshavarzBC,TavliBC}. Of course, multiple strategies exist in this context, but from the scaling law point of view (that is, for large networks), the simplest communication strategy, where source nodes take turns broadcasting their messages to the  entire network, can be shown to be asymptotically optimal (up to logarithmic factors), when the power path loss is that of free space propagation. For a stronger power path loss, still at high SNR, simple multi-hopping strategies also allow to achieve an asymptotically optimal broadcast capacity, so there is not much to be discussed either in this case from the scaling law point of view.

In the present paper, we address the low SNR regime and consider the line-of-sight (LOS) propagation model described in Section \ref{sec:model} below. In this regime, the power available does not allow for a source node to successfully transmit a message to its nearest neighbour without waiting for some amount of time in order to spare power. In this case, %repetition
contrary to the high SNR case, none of the two strategies described above (time-division or multi-hop broadcasting) is asymptotically optimal. This issue was first revealed in \cite{AllaHierBF} in the context of one-dimensional networks, under the LOS model. For such networks, the authors proposed a hierarchical beamforming scheme to broadcast data to the network, that was proven to achieve asymptotic optimal performance.

The generalization of this idea to two-dimensional networks is not immediate. Indeed, a particular feature of one-dimensional networks is that it is always possible for a group of nodes to beamform a given signal to all the other nodes in the network {\em simultaneously}. In two dimensions, a full beamforming gain is only achievable between groups of nodes that are sufficiently far apart from each other. This was already observed in \cite{AllaTelescopBF}, where a strategy was developed to enhance multiple-unicast communications in wireless networks under the LOS model. Taking inspiration from this paper, we propose below a new multi-stage beamforming scheme which is shown to achieve asymptotic optimal performance for broadcasting information in a two-dimensional wireless network.

An interesting aspect of our broadcast strategy is that it achieves the same performance as plain time-division, but with asymptotically much less power. In a large network, this could allow for example to send control signals or channel state information at low cost in the network, without hurting other transmissions. 

We give a detailed description of the scheme in Section \ref{sec:scheme}, as well as a proof of its optimality in Section \ref{sec:opt}. The proof of optimality is done in two steps. We first provide a general upper bound on the broadcast capacity of wireless networks (see Theorem \ref{thm:cap}), whose expression involves the matrix made of fading coefficients between the nodes in the network. We then proceed to characterize the broadcast capacity of two-dimensional wireless networks under the LOS model, by obtaining an asymptotic upper bound on the largest singular value of the above mentioned matrix. This result is of interest in its own right, as such matrices have not been previously studied in the mathematical literature. In particular, there is much less randomness in such a matrix than in classically studied random matrices. We propose here a recursive method to upper bound its largest singular value.

%%%%%%%%%%%%%%%%%%%%%%%%
%%%%%%%%%%%%%%%%%%%%%%%%
%%%%%%%%%%%%%%%%%%%%%%%%
\section{Model} \label{sec:model}
There are $n$ nodes uniformly and independently distributed in a square of area $A=n$, so that the node density remains constant as $n$ increases.  Every node wants to broadcast a different message to the whole network, and all nodes want to communicate at a common {\em per user} data rate $r_n$ bits/s/Hz. We denote by $R_n = n \, r_n$ the resulting {\em aggregate} data rate and will often refer to it simply as ``broadcast rate'' in the sequel. The broadcast capacity of the network, denoted as $C_n$, is defined as the maximum achievable aggregate data rate $R_n$. We assume that communication takes place over a flat channel with bandwidth $W$ and that the signal $Y_j[m]$ received by the $j$-th node at time $m$ is given by
$$
Y_j[m] = \sum_{k \in {\mathcal T}} h_{jk} \, X_k[m] + Z_j[m],
$$
where $\mathcal T$ is the set of transmitting nodes, $X_k [m]$ is the signal sent at time $m$ by node $k$ and $Z_j [m]$ is additive white circularly symmetric Gaussian noise (AWGN) of power spectral density $N_0/2$ Watts/Hz. We also assume a common average power budget per node of $P$ Watts, which implies that the signal $X_k$ sent by node $k$ is subject to an average power constraint $\EE(|X_k|^2) \le P$. In line-of-sight environment, the complex baseband-equivalent channel gain $h_{jk}$ between transmit node $k$ and receive node $j$ is given by
\begin{equation} \label{eq:model}
h_{jk} = \sqrt{G} \; \dfrac{\exp(2 \pi i r_{jk} / \lambda)}{r_{jk}},
\end{equation}
where $G$ is Friis' constant, $\lambda$ is the carrier wavelength, and $r_{jk}$ is the distance between node $k$ and node $j$. Let us finally define
$$
\SNR_s=\frac{GP}{N_0 W},
$$
which is the SNR available for a communication between two nodes at distance $1$ in the network.

It should be noticed that the above line-of-sight model departs from the traditional assumption of i.i.d.~phase shifts in wireless networks. The latter assumption is usually justified by the fact that inter-node distances are in practice much larger than the carrier wavelength, implying that the numbers $2 \pi r_{jk} / \lambda$ can be roughly considered as i.i.d. This approximation was however shown in \cite{FMM09} to be inaccurate in the setting considered in the present paper. A second remark is that no multipath fading is consdered here, which would probably reduce in practice the efficiency of the strategy proposed in the following paragraph.

We focus in the following on the low SNR regime, by which we mean, as in \cite{AllaHierBF}, that $\SNR_s=n^{-\gamma}$ for some constant $\gamma>0$. This means that the power available at each node does not allow for a constant rate direct communcation with a neighbor. This could be the case e.g., in a sensor network with low battery nodes, or in a sparse network with long distances between neighboring nodes.

In order to simplify notation, we choose new measurement units such that $\lambda=1$ and $G/(N_0 W)=1$ in these units. This allows us to write in particular that $\SNR_s=P$.

%%%%%%%%%%%%%%%%%%%%%%%%
%%%%%%%%%%%%%%%%%%%%%%%%
%%%%%%%%%%%%%%%%%%%%%%%%
\section{Back-and-Forth Beamforming Strategy} \label{sec:scheme}
First note that under the LOS model \eqref{eq:model} and the assumptions made in the previous section, the time division scheme described in the introduction achieves a broadcast (aggregate) rate $R_n$ of order $\min(P,1)$. Indeed, a rate of order $1$ is obviously achieved at high SNR\footnote{We coarsely approximate $\log P$  by $1$ here!}. At low SNR (i.e.~when $P \sim n^{-\gamma}$ for some $\gamma>0$), each node can spare power while the others are transmitting, so as to compensate for the path loss of order $1/n$ between the source node and other nodes located at distance at most $\sqrt{2n}$, leading to a broadcast rate of order $R_n \sim \log(1+ n P/n) \sim P$. As we will see, this broadcast rate is not optimal at low SNR. % change?

In the following, we propose a new broadcasting scheme that will prove to be order-optimal. In this new scheme, source nodes still take turns broadcasting their messages, but each transmission is followed by a series of network-wide back-and-forth transmissions that reinforce the strength of the signal, so that at the end, every node is able to decode the message sent from the source. The reason why back-and-forth transmissions are useful here is that in line-of-sight environment, nodes are able to (partly) align the transmitted signals so as to create a significant beamforming gain for each transmission (whereas this would not be the case in high scattering environment with i.i.d.~fading coefficients).\\

{\bf Scheme Description.} The scheme is split into two phases:\\

\textbf{Phase 1. Broadcast Transmission.} The source node broadcasts its message to the whole network. All the nodes receive a noisy version of the signal in this phase, which remains undecoded. This phase only requires one time slot.\\

\textbf{Phase 2. Back-and-Forth Beamforming with Time Division.} Let us first present here an idealized version of this second phase: upon receiving the signal from the broadcasting node, nodes start multiple back-and-forth beamforming transmissions between the two halves of the network, in order to enhance the strength of the signal. Although this simple scheme probably achieves the optimal performance claimed in Theorem \ref{thm:BF_TDMA} below, we lack the analytical tools to prove it. We therefore propose a time-division strategy, where clusters of size $M=\frac{n^{1/4}}{2c_1}\times \frac{n^{1/2}}{4}$ and separated by horizontal distance $d=\frac{n^{1/2}}{4}$ pair up for the back-and-forth transmissions, as illustrated on Fig.~\ref{fig:1}. During each transmission, there are $\Theta\left(n^{1/4-\epsilon}\right)$ cluster pairs operating in parallel (see below), so $\Theta(n^{1-\epsilon})$ nodes are communicating in total. The number of rounds needed to serve all nodes must therefore be $\Theta(n^{\epsilon})$.

\begin{figure}[t]
\centering
\includegraphics[scale=0.5]{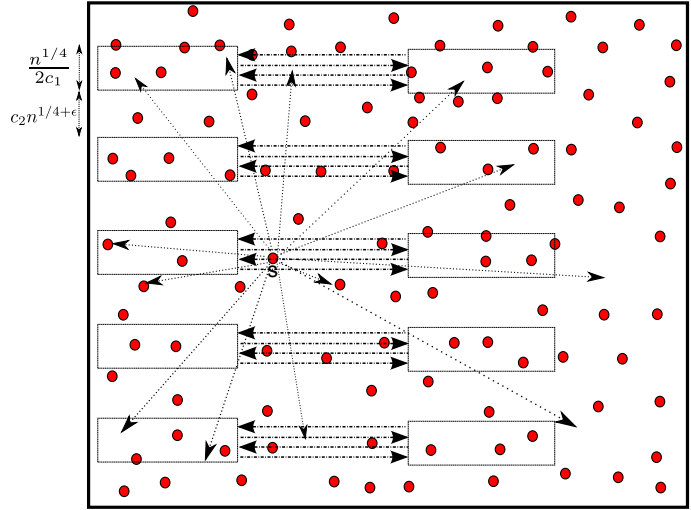}
\caption{$\sqrt{n}\times \sqrt{n}$ network divided into clusters of size $M=\frac{n^{1/4}}{2c_1}\times \frac{n^{1/2}}{4}$.
Two clusters of size $M$ placed on the same horizontal line and separated by distance $d=\frac{{n^{1/2}}}{4}$ pair up and start back-and-forth beamforming. The vertical separation between adjacent cluster pairs is $c_2n^{1/4+\epsilon}$.}
\label{fig:1}
\end{figure}  

After each transmission, the signal received by a node in a given cluster is the sum of the signals coming from the facing cluster, of those coming from other clusters, and of the noise. We assume a sufficiently large vertical distance $c_2 n ^{1/4+\epsilon}$ separating any two cluster pairs, as illustrated on Fig.~\ref{fig:1}. We show below that the broadcast rate between the operating clusters is $\Theta(n^{\frac{1}{2}}P)$. %????
Since we only need $\Theta(n^{\epsilon})$ number of rounds to serve all clusters, phase 2 requires $\Theta(n^{-\frac{1}{2}+\epsilon}P^{-1})$ time slots. As such, back-and-forth beamforming achieves a broadcast rate of $\Theta(n^{\frac{1}{2}-\epsilon}P)$ bits per time slot.\\

In view of the described scheme, we are able to state the following result.

\begin{thm} \label{thm:BF_TDMA}
For any $\epsilon>0$ and $P=O(n^{-\frac{1}{2}})$, the following broadcast rate
\begin{align*}
R_n=\Omega\left( n^{\frac{1}{2}-\epsilon} P\right)
\end{align*}
is achievable with high probability\footnote{that is, with probability at least $1-O\left(\frac{1}{n^p}\right)$ as $n \to \infty$, where the exponent $p$ is as large as we want.} in the network. As a consequence, when $P = \Omega(n^{-\frac{1}{2}})$, a broadcast rate $R_n = \Omega(n^{-\epsilon})$ is achievable with high probability.
\end{thm}

Before proceeding with the proof of the theorem, the following lemma provides an upper bound on the probability that the number of nodes inside each cluster deviates from its mean by a large factor. Its proof can be found in \cite{AllaThesis}, but is also provided in the Appendix for completeness.

\begin{lem} \label{lem:cluster_size}
Let us consider a cluster of area $M$ with $M=n^\beta$ for some $0 < \beta <1$. The number of nodes inside each cluster is then between $((1-\delta)M,\,(1+\delta)M)$ with probability larger than $1-\frac{n}{M}\exp(-\Delta(\delta)M)$ where $\Delta(\delta)$ is independent of $n$ and satisfies $\Delta(\delta)>0$ for $\delta>0$.
\end{lem}

As shown in Fig. \ref{fig:1}, two clusters of size $M=\frac{n^{1/4}}{2c_1}\times \frac{n^{1/2}}{4}$ placed on the same horizontal line and separated by distance $d=\frac{{n^{1/2}}}{4}$ form a cluster pair. During the back-and-forth beamforming phase, there are many cluster pairs operating simultaneously. Given that the cluster width is $\frac{n^{1/4}}{2 c_1}$ and the vertical separation between adjacent cluster pairs is $c_2n^{1/4+\epsilon}$, there are
$$
N_C=\frac{n^{1/2}}{\frac{n^{1/4}}{2 c_1}+c_2n^{1/4+\epsilon}} = \Theta \left( n^{1/4-\epsilon} \right)
$$
cluster pairs operating at the same time. Let $\mathcal{R}_i$ and $\mathcal{T}_i$ denote the receiving and the transmitting clusters of the $i$-th cluster pair, respectively.

Two key ingredients for analyzing the multi-stage back-and-forth beamforming scheme are given in Lemma \ref{lem:Full_Beamforming} and Lemma \ref{lem:Interference}. The proofs are presented in the Appendix. 

\begin{lem} \label{lem:Full_Beamforming}
The maximum beamforming gain between the two clusters of the $i$-th cluster pair can be achieved by using a compensation of the phase shifts at the transmit side which is proportional to the horizontal positions of the nodes. More precisely, there exist a constant $c_1>0$ (remember that $c_1$ is inversely proportional to the width of cluster $i$) and a constant $K_1>0$ such that the magnitude of the received signal at node $j \in \mathcal{R}_i$ is lower bounded with high probability by
$$
\left|\sum_{k\in \mathcal{T}_i} \frac{\exp(2\pi i (r_{jk}-x_k))}{r_{jk}}\right| \ge K_1 \frac{M}{d},
$$
where $x_k$ denotes the horizontal position of node $k$.
\end{lem} 

\begin{lem} \label{lem:Interference}
For every constant $K_2>0$, there exists a sufficiently large separating constant $c_2>0$ such that the magnitude of interfering signals from the simultaneously operating cluster pairs at node $j\in\mathcal{R}_i$ is upper bounded with high probability by
$$
\left|\sum_{\substack{l=1\\ l\neq i}}^{N_C}\sum_{k\in \mathcal{T}_l} \frac{\exp(2\pi i (r_{jk}-x_k))}{r_{jk}}\right| \le K_2\,\frac{M}{d \, n^{\epsilon}}\,\log n.
$$
\end{lem}

\begin{proof}[Proof of Theorem \ref{thm:BF_TDMA}]
The first phase of the scheme results in noisy observations of the message $X$ at all nodes, which are given by
\begin{align*}
Y_k^{(0)}=\sqrt{\SNR_k}\,X+Z_k^{(0)},
\end{align*}
where $\EE(|X|^2)=\EE(|Z_k^{(0)}|^2)=1$ and $\SNR_k$ is the signal-to-noise ratio of the signal $Y_k^{(0)}$ received at the $k$-th node. In what follows, we drop the index $k$ from $\SNR_k$ and only write $\SNR=\min_k\{\SNR_k\}$. Note that it does not make a difference at which side of the cluster pairs the back-and-forth beamforming starts or ends. Hence, assume the left-hand side clusters ignite the scheme by amplifying and forwarding the noisy observations of $X$ to the right-hand side clusters. The signal received at node $j\in\mathcal{R}_i$ is given by
\begin{equation}\label{Yj:1}
Y_j^{(1)} = \sum_{l=1}^{N_C}\sum_{k\in\mathcal{T}_l} \frac{\exp(2\pi i (r_{jk}-x_k))}{r_{jk}} A Y_k^{(0)} + Z_j^{(1)}
\end{equation}
where $A$ is the amplification factor (to be calculated later) and $Z_j^{(1)}$ is additive white Gaussian noise of variance $\Theta(1)$. We start by applying Lemma \ref{lem:Full_Beamforming} and Lemma \ref{lem:Interference} to lower bound
\begin{align*}
\left|\sum_{l=1}^{N_C}\sum_{k\in\mathcal{T}_l} \frac{\exp(2\pi i (r_{jk}-x_k))}{r_{jk}}\right|&\geq \left|\sum_{k\in\mathcal{T}_i} \frac{\exp(2\pi i (r_{jk}-x_k))}{r_{jk}}\right| - \left|\sum_{\substack{l=1\\ l\ne i}}^{N_C}\sum_{k\in\mathcal{T}_l} \frac{\exp(2\pi i (r_{jk}-x_k))}{r_{jk}}\right| \\
&\ge \left(K_1-K_2\frac{\log n}{n^{\epsilon}}\right) \frac{M}{d}=\Theta\left(\frac{M}{d}\right).
\end{align*}
For the sake of clarity, we can therefore approximate\footnote{We make this approximation to lighten the notation and make the exposition clear, but needless to say, the whole analysis goes through without the approximation; it just becomes barely readable.} the expression in \eqref{Yj:1} as follows
\begin{align*}
Y_j^{(1)} & = \sum_{l=1}^{N_C}\sum_{k\in\mathcal{T}_l} \frac{\exp(2\pi i (r_{jk}-x_k))}{r_{jk}} A \sqrt{\SNR_k}\,X + \sum_{l=1}^{N_C}\sum_{k\in\mathcal{T}_l} \frac{\exp(2\pi i (r_{jk}-x_k))}{r_{jk}} A Z_k^{(0)} + Z_j^{(1)}\\
& \simeq\frac{A M}{d} \sqrt{\SNR}\,X + \frac{A\sqrt{N_C M}}{d} Z^{(0)} + Z_j^{(1)} =\frac{A M}{d} \sqrt{\SNR}\,X + \frac{A M}{d}\sqrt{\frac{N_C}{M}} Z^{(0)} + Z_j^{(1)},
\end{align*}
where
$$
Z^{(0)}=\frac{d}{\sqrt{N_C  M}} \sum_{l=1}^{N_C} \sum_{k\in\mathcal{T}_l} \frac{\exp(2\pi i (r_{jk}-x_k))}{r_{jk}} Z_k^{(0)}.
$$
Note that $\EE(|Z^{(0)}|^2)=\Theta(1)$. Repeating the same process $t$ times in a back-and-forth manner results in a final signal at node $j\in\mathcal{R}_i$ in the left or the right cluster (depending on whether $t$ is odd or even) that is given by 
\begin{align*}
Y_j^{(k)} & = \left(\frac{A M}{d}\right)^t \sqrt{\SNR}\,X + \left(\frac{A M}{d}\right)^{t}\sqrt{\frac{N_C}{M}} \, Z^{(0)}\\
&  + \ldots + \left(\frac{A M}{d}\right)^{t-s}\sqrt{\frac{N_C}{M}} \, Z^{(s)}+\ldots+Z_j^{(t)},
\end{align*}
where
$$
Z^{(s)}= \frac{d}{\sqrt{N_C M}}\sum_{b=1}^{N_C}\sum_{k\in\mathcal{T}_b} \frac{\exp(2\pi i (r_{jk}-x_k))}{r_{jk}} Z_k^{(s)}.
$$
Note again that $\EE(|Z^{(s)}|^2)=\Theta(1)$, and $Z_j^{(t)}$ is additive white Gaussian noise of variance $\Theta(1)$. Finally, note that Lemma \ref{lem:Interference} ensures an upper bound on the beamforming gain of the noise signals, i.e.,  
$$
\left|\sum_{l=1}^{N_C}\sum_{k\in\mathcal{T}_l} \frac{\exp(2\pi i (r_{jk}-x_k))}{r_{jk}}\right|
\le \left|\sum_{k\in\mathcal{T}_i} \frac{\exp(2\pi i (r_{jk}-x_k))}{r_{jk}}\right| + \left|\sum_{\substack{l=1\\ l\ne i}}^{N_C}\sum_{k\in\mathcal{T}_l} \frac{\exp(2\pi i (r_{jk}-x_k))}{r_{jk}}\right| \le \left(1+K_2\frac{\log n}{n^{\epsilon}}\right)\frac{M}{d}.
$$
(notice indeed that the first term in the middle expression is trivially upper bounded by $M/d$, as it contains $M$ terms, all less than $1/d$). Now, we want the power of the signal to be of order 1, that is:
\begin{align}\label{eq:signal_power}
& \EE\left(\left(\left(\frac{A M}{d}\right)^t \sqrt{\SNR}\,X\right)^2\right)
= \left(\frac{A M}{d}\right)^{2t} \SNR = \Theta(1)\\\nonumber
& \Rightarrow A = \Theta\left(\frac{d}{M}\,\SNR^{-\frac{1}{2t}}\right).
\end{align}
Since at each round of TDMA cycle there are $\Theta(N_C M)$ nodes transmitting, then every node will be active $\Theta\left(\frac{N_C M}{n}\right)$ fraction of the time. As such, the amplification factor is given by
$$
A=\Theta\left(\sqrt{\frac{n}{N_C M}\tau P}\right),
$$
where $\tau$ is the number of time slots between two consecutive transmissions, i.e. every $\tau$ time slots we have one transmission. Therefore, we have
\begin{align*}
&A = \Theta\left(\frac{d}{M}\,\SNR^{-\frac{1}{2t}}\right)=\Theta\left(\sqrt{\frac{n}{N_C M}\tau P}\right)\\
& \Rightarrow \tau = \Theta\left(\frac{N_C \, d^2}{n \, M \, P}\,\SNR^{-\frac{1}{t}}\right).
\end{align*}
We can pick the number of back-and-forth transmissions $t$ sufficiently large to ensure that $\SNR^{-\frac{1}{t}}=O(n^{\epsilon})$, which results in $$\tau=O\left( \frac{1}{n^{1/2}P}\right).$$  
Moreover, the noise power is given by
\begin{align*}
\sum_{s=0}^{t-1} \EE\left(\left(\left(\frac{A M}{d}\right)^{t-s} \sqrt{\frac{N_C}{M}}Z^{(s)}\right)^2\right) + \EE\left(\left(Z_j^{(t)}\right)^2\right) &\le t \, \EE\left(\left(\left(\frac{A M}{d}\right)^t \sqrt{\frac{N_C}{M}}Z^{(0)}\right)^2\right) + 1\\
&\le t \, \left(\frac{A M}{d}\right)^{2t}{\frac{N_C}{M}}+1\\
&\overset{(a)}{\le}t+1=\Theta(1),
\end{align*}
where $(a)$ is true if and only if $\SNR=\Omega(N_C/M)=\Omega(n^{-1/2-\epsilon})$ (check eq. \eqref{eq:signal_power}), which is true: Distance separating any two nodes in the network is as most $\sqrt{2n}$, which implies that the $\SNR$ of the received signal at all the nodes in the network is $\Omega(n^{-1/2})$.

Given that the required $\tau=O\left(\frac{1}{n^{1/2}P}\right)$, we can see that for $P=O(n^{-1/2})$ the broadcast rate between simultaneously operating clusters is $\Omega(n^{1/2}P)$. Finally, applying TDMA of $\frac{n}{N_C M}=\Theta(n^{\epsilon})$ steps ensures that $X$ is successfully decoded at all nodes and the broadcast rate $R_n=\Omega\left(n^{1/2-\epsilon}P\right)$.

As a last remark, let us mention that the consequence stated in the theorem for the regime where more power is available at the transmitters is an obvious one: by simply reducing the amount of power used at each node to exactly $n^{-1/2} \le P$, one achieves the following broadcast rate, using the first part of the theorem:
$$
R_n = \Omega  \left( n^{\frac{1}{2}-\epsilon} \,  n^{-\frac{1}{2}} \right) = \Omega \left( n^{-\epsilon} \right).
$$
This completes the proof of the theorem. 
\end{proof}

%%%%%%%%%%%%%%%%%%%%%%%%
%%%%%%%%%%%%%%%%%%%%%%%%
%%%%%%%%%%%%%%%%%%%%%%%%
\section{Optimality of the Scheme} \label{sec:opt}
In this section, we first establish a general upper bound on the broadcast capacity of wireless networks at low SNR, which applies to a general fading matrix $H$
(with proper measurement units such that again, $\SNR_s=P$ in these units).

\begin{thm} \label{thm:cap}
Let us consider a network of $n$ nodes and let $H$ be the $n \times n$ matrix with $h_{jj}=0$ on the diagonal and $h_{jk}=$ the fading coefficient between node $j$ and node $k$ in the network. The broadcast capacity of such a network with $n$ nodes is then upper bounded by
$$
C_n \le P \, \Vert H \Vert^2
$$
where $P$ is the power available per node and $\Vert H \Vert$ is the spectral norm (i.e.~the largest singular value) of $H$.
\end{thm}
 
\begin{proof}
Using the classical cut-set bound \cite[Theorem 15.10.1]{CT06}, the following upper bound on the broadcast capacity $C_n$ is obtained:
\begin{align*}
C_n \le \max_{\stackrel{p_X :}{\EE(|X_k|^2) \le P, \; \forall 1 \le k \le n}} \min_{1 \le j \le n} I(X_{\{1,\ldots,n\}\backslash\{j\}} ; Y_j|X_j).
\end{align*}
Moreover, we have
\begin{align*}
I(X_{\{1,\ldots,n\}\backslash\{j\}},X_j ; Y_j)&= I(X_{\{1,\ldots,n\}\backslash\{j\}} ; Y_j)+ I(X_j ; Y_j|X_{\{1,\ldots,n\}\backslash\{j\}})\\
&\overset{(a)}{=} I(X_{\{1,\ldots,n\}\backslash\{j\}} ; Y_j)\\
&=I(X_j ; Y_j)+I(X_{\{1,\ldots,n\}\backslash\{j\}} ; Y_j|X_j)\\
&\overset{(b)}{\geq} I(X_{\{1,\ldots,n\}\backslash\{j\}} ; Y_j|X_j),
\end{align*}
where $(a)$ follows from the fact that $X_j-X_{\{1,\ldots,n\}\backslash\{j\}}-Y_j$ forms a Markov chain, which means that $I(X_j ; Y_j|X_{\{1,\ldots,n\}\backslash\{j\}})=0$, and $(b)$ follows from the fact that $I(X_j ; Y_j)\geq 0$. Therefore, we get
\begin{align*}
C_n &\le \max_{\stackrel{p_X :}{\EE(|X_k|^2) \le P, \; \forall 1 \le k \le n}} \min_{1 \le j \le n} I(X_{\{1,\ldots,n\}\backslash\{j\}} ; Y_j|X_j)\\
& \le \max_{\stackrel{p_X :}{\EE(|X_k|^2) \le P, \; \forall 1 \le k \le n}} \min_{1 \le j \le n} I(X_{\{1,\ldots,n\}\backslash\{j\}} ; Y_j)\\
& \le \max_{\stackrel{Q_X \ge 0}{(Q_X)_{kk} \le P, \; \forall 1 \le k \le n}} \min_{1 \le j \le n} \log (1 + h_j Q_X h_j^{\dagger})
\end{align*}
where $h_j=(h_{j1}, \ldots, h_{j,j-1}, 0, h_{j,j+1}, \ldots, h_{jn})$, as the joint distribution $p_X$ maximizing the above expression is clearly Gaussian. Using then the fact that the minimum of a set of numbers is less than its average, the above expression can be further bounded by 
\begin{align*}
C_n & \le \max_{\stackrel{Q_X \ge 0}{(Q_X)_{kk} \le P, \; \forall 1 \le k \le n}} \frac{1}{n} \sum_{j=1}^n \log (1+h_j Q_X h_j^{\dagger})\\
& = \max_{\stackrel{Q_X \ge 0}{(Q_X)_{kk} \le P, \; \forall 1 \le k \le n}} \frac{1}{n} \sum_{j=1}^n \log \det (I_n + h_j^{\dagger} h_j Q_X)\\
& \le \max_{\stackrel{Q_X \ge 0}{(Q_X)_{kk} \le P, \; \forall 1 \le k \le n}} \log \det \left( I_n + \frac{1}{n} \sum_{j=1}^n h_j^{\dagger} h_j Q_X \right)
\end{align*}
using successively the property that $\log \det(I+AB) = \log \det(I+BA)$ and the fact that $\log \det( \cdot)$ is concave. Observing now that the $n \times n$ matrix $H$ whose entries are given by $h_{jk}= (h_j)_k$ is the one in the theorem statement and that $\sum_{j=1}^n h_j^{\dagger} h_j = H^{\dagger}H$, we can rewrite, using again $\log \det(I+AB) = \log \det(I+BA)$: 
\begin{align*}
C_n &\le \max_{\stackrel{Q_X \ge 0}{(Q_X)_{kk} \le P, \; \forall 1 \le k \le n}} \log \det \left( I_n + \frac{1}{n} H Q_X H^{\dagger}\right)\\
&\le \max_{\stackrel{Q_X \ge 0}{(Q_X)_{kk} \le P, \; \forall 1 \le k \le n}} \frac{1}{n} \, \Tr(H Q_X H^{\dagger}) \\
&\le \max_{\stackrel{Q_X \ge 0}{(Q_X)_{kk} \le P, \; \forall 1 \le k \le n}} \frac{1}{n} \, \Tr(Q_X) \, \Vert H \Vert^2  = P \, \Vert H \Vert^2
\end{align*}
where the last inequality follows from the fact that $\Tr(BAB^{\dagger})\le \Vert B \Vert^2 \Tr(A)$, for any matrix $B$ and $A\geq 0$. This completes the proof.
\end{proof}

We now aim to specialize Theorem \ref{thm:cap} to line-of-sight fading, where the matrix $H$ is given by
\begin{equation} \label{eq:fading_matrix}
h_{jk} = \begin{cases} 0 & \text{if }j=k \\ \dfrac{\exp(2 \pi i r_{jk})}{r_{jk}} & \text{if } j \ne k \end{cases}
\end{equation}
The rest of the section is devoted to proving the proposition below which, together with Theorem \ref{thm:cap}, shows the asymptotic optimality of the back-and-forth beamforming scheme presented in Section \ref{sec:scheme} for two-dimensional networks at low SNR and under LOS fading\footnote{Note that for a one-dimensional network in LOS environment, Theorem \ref{thm:cap} allows to recover the result already obtained in \cite{AllaHierBF}.}.

\begin{figure}[t]
 \centering
 \includegraphics[scale=0.7]{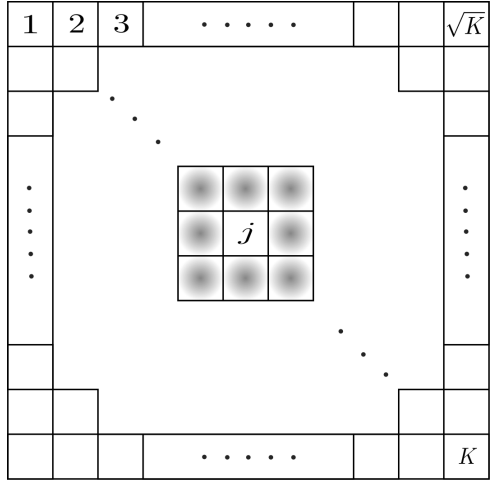}
 \caption{$\sqrt{n}\times\sqrt{n}$ network split into $K$ clusters and numbered in order. As such, $R_j=\{j-\sqrt{K}-1,j-\sqrt{K},j-\sqrt{K}+1,j-1,j,j+1, j+\sqrt{K}-1,j+\sqrt{K},j+\sqrt{K}+1\}$, which represents the center square containing the cluster $j$ and its $8$ neighbors (marked in shades).}
 \label{fig:2}
 \end{figure}
\begin{prop} \label{prop:ub_H}
Let $H$ be the $n \times n$ matrix given by \eqref{eq:fading_matrix}. For every $\varepsilon>0$, there exists a constant $c>0$ such that
$$
\Vert H \Vert^2 \le c \, n^{\frac{1}{2}+\varepsilon}
$$
with high probability as $n$ gets large.
\end{prop}

Analyzing directly the the asymptotic behavior of $\Vert H \Vert$ reveals itself difficult. We therefore decompose our proof into simpler subproblems. The first building block of the proof is the following Lemma, which can be viewed as a generalization of the classical Ger\v{s}gorin discs' inequality.

\begin{lem} \label{lem:Gersgorin}
Let $B$ be an $n \times n$ matrix decomposed into blocks $B_{jk}$, $j,k=1,\ldots,K$, each of size $M \times M$, with $n=K M$. Then
$$
\Vert B \Vert \le \max \left\{ \max_{1 \le j \le K} \sum_{k=1}^K \Vert B_{jk} \Vert, \max_{1 \le j \le K} \sum_{k=1}^K \Vert B_{kj} \Vert \right\}
$$
\end{lem}

The proof of this Lemma is relegated to the Appendix. The second building block of this proof is the following lemma, the proof of which is also given in the Appendix.

\begin{lem}\label{lem:blockNorm}
Let $\widehat{H}$ be the $M \times M$ channel matrix between two square clusters of $M$ nodes distributed uniformly at random, each of area $A=M$. Then there exists a constant $c>0$ such that
$$
\Vert \widehat{H} \Vert^2 \le c \, \frac{M^{1+\epsilon}}{d}
$$
with high probability as $M$ gets large, where $2\sqrt{M} \le d \le M$ denotes the distance between the centers of the two clusters. 
\end{lem}
\begin{proof}[Proof of Proposition \ref{prop:ub_H}]

The strategy for the proof is now the following: in order to bound $\Vert H \Vert$, we divide the matrix into smaller blocks, apply Lemma \ref{lem:Gersgorin} and Lemma \ref{lem:blockNorm} in order to bound the off-diagonal terms $\Vert H_{jk} \Vert$. For the diagonal terms $\Vert H_{jj} \Vert$, we reapply Lemma \ref{lem:Gersgorin} and proceed in a recursive manner, until we reach small size blocks for which a loose estimate is sufficient to conclude.

Let us therefore decompose the network into $K$ clusters of $M$ nodes each, with $n=K M$. By Lemma \ref{lem:Gersgorin}, we obtain
\begin{equation} \label{eq:ub_H}
\Vert H \Vert \le \max \left\{ \max_{1 \le j \le K} \sum_{k=1}^K \Vert H_{jk} \Vert, \max_{1 \le j \le K} \sum_{k=1}^K \Vert H_{kj} \Vert \right\}
\end{equation}
where the $n \times n$ matrix $H$ is decomposed into blocks $H_{jk}$, $j,k=1,\ldots,K$, with $H_{jk}$ denoting the $M \times M$ channel matrix between cluster number $j$ and cluster number $k$ in the network. Let us also denote by $d_{jk}$ the corresponding inter-cluster distance, measured from the centers of these clusters. According to Lemma \ref{lem:blockNorm}, if $d_{jk} \ge 2\sqrt{M}$, then there exists a constant $c>0$ such that
$$
\Vert H_{jk} \Vert^2 \le c \, \frac{M^{1+\epsilon}}{d_{jk}}  \, \le c \, n^{\epsilon} \frac{M}{d_{jk}}
$$
with high probability as $M \to \infty$.

Let us now fix $j \in \{1,\ldots,K\}$ and define $R_j = \{ 1 \le k \le K: d_{jk} < 2\sqrt{M} \}$ and $S_j = \{1 \le k \le K: d_{jk} \ge 2\sqrt{M} \}$ (see Fig. \ref{fig:2}). By the above inequality, we obtain
$$
\sum_{k=1}^K \Vert H_{jk} \Vert \le \sum_{k \in R_j} \Vert H_{jk} \Vert + \sqrt{c \, n^{\epsilon}} \, \sum_{k \in S_j} \sqrt{\frac{M}{d_{jk}}}
$$
with high probability as $M$ gets large. Observe that as there are $8l$ clusters or less at distance $l\sqrt{M}$ from cluster $j$, so we obtain 
\begin{align*}
\sum_{k \in S_j} \sqrt{\frac{M}{d_{jk}}} & \le \sum_{l=2}^{\sqrt{K}} 8l \, \sqrt{\frac{M}{l \sqrt{M}}} = O  \left(M^{1/4} K^{3/4} \right) = O \left( \frac{n^{3/4}}{M^{1/2}} \right) 
\end{align*}
as $K=n/M$. There remains to upper bound the sum over $R_j$. Observe that this sum contains at most 9 terms: namely the term $k=j$ and the 8 terms corresponding to the 8 neighboring clusters of cluster $j$.  It should then be observed that for each $k \in R_j$, $\Vert H_{jk} \Vert \le \Vert H(R_j) \Vert$, where $H(R_j)$ is the $9M \times 9M$ matrix made of the $9 \times 9$ blocks $H_{j_1,j_2}$ such that $j_1,j_2 \in R_j$. Finally, this leads to
$$
\sum_{k=1}^K \Vert H_{jk} \Vert \le 9 \Vert H(R_j) \Vert + \sqrt{c\, n^{\epsilon}} \, \frac{n^{3/4}}{M^{1/2}} 
$$
Using the symmetry of this bound and \eqref{eq:ub_H}, we obtain
\begin{equation} \label{eq:recursion}
\Vert H \Vert \le 9 \, \max_{1 \le j \le K} \Vert H(R_j) \Vert + \sqrt{c\, n^{\epsilon}} \, \frac{n^{3/4}}{M^{1/2}} 
\end{equation}
A key observation is now the following: the $9M \times 9M$ matrix $H(R_j)$ has exactly the same structure as the original matrix $H$. So in order to bound its norm $\Vert H(R_j) \Vert$, the same technique may be reused! This leads to the following recursive Lemma.

\begin{lem} \label{lem:recursion}
Assume there exist constants $c>0$ and $b \in [1/4,1/2]$ such that
\vspace{-1mm}
$$
\Vert H \Vert \le \sqrt{c\, n^{\epsilon}} \; n^b
$$
with high probability as $n$ gets large. Then there exists a constant $c'>0$ such that
\vspace{-1mm}
$$
\Vert H \Vert \le \sqrt{c'\, n^{\epsilon}} \; n^{f(b)} 
$$
with high probability as $n$ gets large, where $f(b)=\frac{3b}{4b+2}<b$.
\end{lem}

\begin{proof}
The assumption made implies that there exist $c>0$ and $b \in [1/4,1/2]$ such that for every $M \times M$ diagonal subblock $H_M$ of the matrix $H$, \vspace{-1mm}
$$
\Vert H_M \Vert \le \sqrt{c\, M^{\epsilon}} \; M^b \le \sqrt{c\, n^{\epsilon}} \; M^b
$$
with high probability as $M$ gets large. Together with \eqref{eq:recursion}, this implies that
\begin{align*}
\Vert H \Vert & \le 9 \, \sqrt{c\, n^{\epsilon}} \; M^b + \sqrt{c\, n^{\epsilon}} \; \frac{n^{3/4}}{M^{1/2}}\\
& = 10\,\sqrt{c\, n^{\epsilon}} \; \left( M^b + \frac{n^{3/4}}{M^{1/2}} \right)
\end{align*}
Choosing $M = \lfloor n^{3/(4b+2)} \rfloor$, we obtain
$$
\Vert H \Vert \le \sqrt{c'\, n^{\epsilon}} \; n^{3b/(4b+2)}.
$$
\end{proof}

Besides, it is easy to check that the assumption of Lemma \ref{lem:recursion} holds with $b=1/2$. Apply for this the slightly modified version of the classical Ger\v{s}gorin inequality (which is nothing but the statement of Lemma \ref{lem:Gersgorin} applied to the case $M=1$):
$$
\Vert H \Vert \le \max\left\{ \max_{1 \le j \le n} \sum_{k=1}^n |h_{jk}|, \max_{1 \le j \le n} \sum_{k=1}^n |h_{kj}| \right\} = \max_{1 \le j \le n} \sum_{k=1 \atop k \ne j}^n \frac{1}{r_{jk}}
$$
For any $1 \le j \le n$, it holds with high probability that for $c$ large enough, 
$$
\sum_{k=1 \atop k \ne j}^n \frac{1}{r_{jk}} \le \sum_{l=1}^{\sqrt{n}} (cl \log n) \, \frac{1}{l} = O(\sqrt{n} \log n)
$$
which implies that $\Vert H \Vert = O  \left( \sqrt{n^{1+\epsilon}} \right)$ for any $\epsilon>0$.

By applying Lemma \ref{lem:recursion} successively, we obtain a decreasing sequence of upper bounds on $\Vert H \Vert$:
$$
\Vert H \Vert \le \sqrt{c\, n^{\epsilon}} \; n^{b_0}, \quad \le \sqrt{c\, n^{\epsilon}} \; n^{b_1}, \quad \le \sqrt{c\, n^{\epsilon}} \; n^{b_2}
$$
where the sequence $b_0=1/2$, $b_1=f(b_0)=3b_0/(4b_0+2)=3/8$, $b_2=f(b_1)=3b_1/(4b_1+2)=9/28$ converges to the fixed point $b^*=f(b^*)=1/4$ (as $f$ is strictly increasing on $[\frac{1}{4},\frac{1}{2}]$ and $f(b) <b$ for every $\frac{1}{4} < b \le \frac{1}{2}$). This finally proves Proposition \ref{prop:ub_H}. \end{proof}

%%%%%%%%%%%%%%%%%%%%%%%%
%%%%%%%%%%%%%%%%%%%%%%%%
%%%%%%%%%%%%%%%%%%%%%%%%
\section{Conclusion}
In this work, we characterize the broadcast capacity of two-dimensional wireless networks at low SNR in line-of-sight environment, which is achieved via a back-and-forth beamforming scheme. We showed that the broadcast capacity is upper bounded by the total power transfer in the network, which in turn is equal to $P \, \Vert H \Vert^2$. We present a detailed analysis of the largest singular value of the fading matrix $H$. We further present a practical broadcasting scheme that guarantees the total power transfer throughout the network. This scheme relies on back-and-forth beamforming among clusters through multiple stage time division channel accesses.

\section{Acknowledgment}
S. Haddad's work is supported by Swiss NSF Grant Nr.~200020-156669.

%%%%%%%%%%%%%%%%%%%%%%%%
%%%%%%%%%%%%%%%%%%%%%%%%
%%%%%%%%%%%%%%%%%%%%%%%%
\appendix 

\begin{proof}[Proof of Lemma \ref{lem:cluster_size}]
The number of nodes in a given cluster is the sum of $n$ independently and identically distributed Bernoulli random variables $B_i$, with $\mathcal{P}(B_i=1)=M/n$. Hence
\begin{align*}
& \mathbb{P}\left(\sum_{i=1}^n B_i\geq (1+\delta)M\right)\\
& =\mathbb{P}\left(\exp\left(s\sum_{i=1}^n B_i\right)\geq \exp(s(1+\delta)M)\right)\\
& \leq \mathbb{E}^n(\exp(sB_1))\exp(-s(1+\delta)M)\\
& = \left(\frac{M}{n}\exp(s)+1-\frac{M}{n}\right)^n\exp(-s(1+\delta)M)\\
& \leq \exp(-M(s(1+\delta)-\exp(s)+1))=\exp(-M\Delta_+(\delta))
\end{align*}
where $\Delta_+(\delta)=(1+\delta)\log(1+\delta)-\delta$ by choosing $s=\log(1+\delta)$. The proof of the lower bound follows similarly by considering the random variables $-B_i$. The conclusion follows from the union bound.
\end{proof}
\begin{figure}
\centering
\includegraphics[scale=0.7]{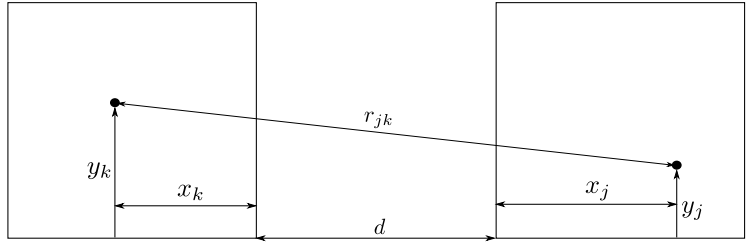}
\caption{Coordinate system.}
\label{fig:3}
\end{figure}  
\begin{proof}[Proof of Lemma \ref{lem:Full_Beamforming}]
We present lower and upper bounds on the distance $r_{jk}$ separating a receiving node $j\in \mathcal{R}_i$ and a transmitting node $k\in\mathcal{T}_i$. Denote by $x_j$, $x_k$, $y_j$, and $y_k$ the horizontal and the vertical positions of nodes $j$ and $k$, respectively (as shown in Fig. \ref{fig:3}). An easy lower bound on $r_{jk}$ is
$$
r_{jk}\geq x_k+x_j+d
$$
On the other hand, using the inequality $\sqrt{1+x} \le 1 + \frac{x}{2}$, we obtain
\begin{align*}
r_{jk}&=\sqrt{(x_k+x_j+d)^2+(y_j-y_k)^2}\\
&=\left(x_k+x_j+d\right)\sqrt{1+\frac{(y_j-y_k)^2}{(x_k+x_j+d)^2}}\\
&\leq x_k+x_j+d+\frac{(y_j-y_k)^2}{2d}\leq x_k+x_j+d+\frac{1}{2c_1^2}.
\end{align*}
Therefore, $$0\leq r_{jk}-x_k-x_j-d \leq \frac{1}{2c_1^2}.$$ After bounding $r_{jk}$, we can proceed to the proof of the lemma as follows:
\begin{align*}
\left|\sum_{k\in \mathcal{T}_i} \frac{\exp(2\pi i (r_{jk}-x_k))}{r_{jk}}\right|&=  \left|\sum_{k\in \mathcal{T}_i} \frac{\exp(2\pi i (r_{jk}-x_k-x_j-d))}{r_{jk}}\right|\\
&\geq \Re \left(\sum_{k\in \mathcal{T}_i} \frac{\exp(2\pi i (r_{jk}-x_k-x_j-d))}{r_{jk}}\right)\\
&\geq \sum_{k\in \mathcal{T}_i} \frac{\cos\left(\frac{\pi}{ c_1^2}\right)}{r_{jk}}\ge K_1\frac{M}{d},
\end{align*}
when the constant $c_1$ is chosen sufficiently large so that $\cos\left(\frac{\pi}{ c_1^2}\right)>0$.
\end{proof}

\begin{proof}[Proof of Lemma \ref{lem:Interference}]
There are $N_C$ clusters transmitting simultaneously. Except for the horizontally adjacent cluster of a given cluster pair ($i$-th cluster pair), all the rest of the transmitting clusters are considered as interfering clusters (there are $N_C-1$ of these). With high probability, each cluster contains $\Theta(M)$ nodes. For the sake of clarity, we assume here that every cluster contains exactly $M$ nodes, but the argument holds in the general case. In this lemma, we upper bound the magnitude of interfering signals from the simultaneously interfering clusters at node $j\in\mathcal{R}_i$ as follows
\begin{align*}
\left|\sum_{\substack{l=1\\ l\neq i}}^{N_C}\sum_{k\in \mathcal{T}_l} \frac{\exp(2\pi i (r_{jk}-x_k))}{r_{jk}}\right| & \le \sum_{\substack{l=1\\ l\neq i}}^{N_C}\left|\sum_{k\in \mathcal{T}_l} \frac{\exp(2\pi (r_{jk}-x_k))}{r_{jk}}\right|\\
&\le \sum_{\substack{l=1\\ l\neq i}}^{N_C}\left|\sum_{k\in \mathcal{T}_l} \frac{\cos(2\pi (r_{jk}-x_k))}{r_{jk}}\right|+\sum_{\substack{l=1\\ l\neq i}}^{N_C}\left|\sum_{k\in \mathcal{T}_l} \frac{\sin(2\pi (r_{jk}-x_k))}{r_{jk}}\right|\\
&\le 2\sum_{\substack{l=1}}^{N_C}\left|\sum_{k\in \mathcal{T}'_l} \frac{\cos(2\pi (r_{jk}-x_k))}{r_{jk}}\right|+2\sum_{\substack{l=1}}^{N_C}\left|\sum_{k\in \mathcal{T}'_l} \frac{\sin(2\pi (r_{jk}-x_k))}{r_{jk}}\right|
\end{align*}
where $\mathcal{T}'_l$ denotes the $l$-th interfering transmit cluster that is at a vertical distance of $l \left(\frac{n^{1/4}}{2c_1}+c_2 n^{1/4+\epsilon}\right)$ from the desired receiving cluster $\mathcal{R}_i$. We further upper bound the first term (cosine terms) in the equation above as follows (notice that we can upper bound the second term (sine terms) in exactly the same fashion):
\begin{align}\nonumber
\left|\sum_{k\in \mathcal{T}'_l} \frac{\cos(2\pi (r_{jk}-x_k))}{r_{jk}}\right|&=\left|\sum_{k\in \mathcal{T}'_l} X^{(l)}_k\right|\\\nonumber
&=\left|\sum_{k\in \mathcal{T}'_l}\left(X_k-\mathbb{E} \left(X^{(l)}_k\right)\right)+\sum_{k\in \mathcal{T}'_l}\mathbb{E} \left(X^{(l)}_k
\right)\right| \\\nonumber
&\overset{(a)}{\le} \left|\sum_{k\in \mathcal{T}'_l}\left(X^{(l)}_k-\mathbb{E} \left(X^{(l)}_k\right)\right)\right|+\left|\sum_{k\in \mathcal{T}'_l}\mathbb{E} \left(X^{(l)}_k\right)\right|\\\label{eq:Expectation_X1}
&\overset{(b)}{=} M \left|\frac{1}{M} \sum_{k\in \mathcal{T}'_l} \left(X^{(l)}_k-\mathbb{E} \left(X^{(l)}_k\right)\right) \right|+M \left|\mathbb{E} \left(X^{(l)}_1\right)\right|
\end{align}
where $(a)$ follows from the triangle inequality and $(b)$ results from the fact that the $X^{(l)}_k$'s (note that $X^{(l)}_k=(\cos(2\pi (r_{jk}-x_k)))/(r_{jk})\,\,\forall k\in\mathcal{T}'_l$) are independent and identically distributed. Let us first bound the second term of \eqref{eq:Expectation_X1}: $\forall k\in\mathcal{T}'_l$, we have
$$
|r_{jk}| = r_{jk} = \sqrt{(x_k+x_j+d)^2+(y_j-y_k)^2} \ge d = \frac{n^{1/2}}{4}$$
is a $C^2$ function and 
\begin{align*}
|r_{jk}'(y_k)|&=\left|\frac{\partial\,r_{jk} }{\partial{y_k}}\right| = \frac{|y_k-y_j|}{r_{jk}} \\
&\ge \frac{l \, c_2 \, n^{1/4+\epsilon}+(l-1) \, \frac{n^{1/4}}{2c_1}}{n^{1/2}}\\
&\ge {l \, c_2 \, n^{-1/4+\epsilon}}
\end{align*}
Moreover, $r_{jk}''$ changes sign at most twice. By the integration by parts formula, we obtain
\begin{align*}
\int_{{y_k}_0}^{{y_k}_1} \,dy_k \frac{\cos(2\pi  r_{jk})}{r_{jk}} & =
\int_{{y_k}_0}^{{y_k}_1} dy_k \, \frac{2 \pi  r'_{jk}}{2 \pi   r'_{jk}r_{jk}} \, \cos( 2\pi   r_{jk})\\
&=  \frac{ -\sin( 2 \pi  r_{jk})}{2 \pi   r'_{jk}r_{jk}} \bigg|_{{y_k}_0}^{{y_k}_1} + \frac{1}{2\pi }\int_{{y_k}_0}^{{y_k}_1} dy_k \, \frac{r_{jk}r''_{jk}+(r'_{jk})^2}{(r'_{jk}r_{jk})^2} \, \sin({ 2 \pi  r_{jk}})
\end{align*}
which in turn yields the upper bound
\begin{align*}
\left|\int_{{y_k}_0}^{{y_k}_1} dy_k \, \frac{\cos({2 \pi  r_{jk}})}{r_{jk}} \right|&\le \frac{1}{2 \pi} \, \Bigg( \frac{2}{\min_{y_k}\{|r'_{jk}||r_{jk}|\}} 
+ \int_{{y_k}_0}^{{y_k}_1} dy_k \, \frac{|r''_{jk}|}{(r'_{jk})^2|r_{jk}|} + \int_{{y_k}_0}^{{y_k}_1} dy_k \, \frac{1}{r_{jk}^2} \Bigg)\\
&\le \frac{1}{2 \pi} \, \Bigg( \frac{4}{l \, c_2 \, n^{1/4+\epsilon}} 
+ \frac{1}{\min_{y_k}\{|r_{jk}|\}}\int_{{y_k}_0}^{{y_k}_1} dy_k \, \frac{|r''_{jk}|}{(r'_{jk})^2} + \frac{|{{y_k}_1}-{{y_k}_0}|}{\min_{y_k}\{r_{jk}^2\}} \Bigg)\\
&\le \frac{1}{2 \pi} \, \Bigg( \frac{4}{l \, c_2 \, n^{1/4+\epsilon}} 
+ \frac{4}{l \, c_2 \, n^{1/4+\epsilon}} + \frac{2}{n^{3/4}} \Bigg)
\le \frac{9/(2\pi)}{l \, c_2 \, n^{1/4+\epsilon}}.
\end{align*}
Therefore, for any $k\in\mathcal{T}'_l$,
\begin{align}\nonumber
\bigg|\mathbb{E}\left(X^{(l)}_k\right)\bigg|&=\left|\frac{4}{n^{1/2}}\int_0^{\frac{n^{1/2}}{4}}\,dx_k\frac{1}{|{y_k}_1-{y_k}_0|}\int_{{y_k}_0}^{{y_k}_1} dy_k \, \frac{\cos({2 \pi  r_{jk}})}{r_{jk}} \right|\\\nonumber
&\le \frac{4}{n^{1/2} \, |{y_k}_1-{y_k}_0|}\int_0^{\frac{n^{1/2}}{4}}\,dx_k \left|\int_{{y_k}_0}^{{y_k}_1} dy_k \, \frac{\cos({2 \pi  r_{jk}})}{r_{jk}} \right|\\
&\le \frac{9/(2\pi)}{|{y_k}_1-{y_k}_0| \, l \, c_2 \, n^{1/4+\epsilon}}
\le \frac{9c_1}{\pi c_2}\frac{1}{l \, n^{1/2+\epsilon}}=\frac{9c_1}{\pi c_2} \frac{1}{l \, d \, n^{\epsilon}}.\label{eq:Exp_1}
\end{align}
We further upper bound the first term in \eqref{eq:Expectation_X1} by using the Hoeffding's inequality \cite{Hoeffding}. Note that the $X^{(l)}_k$'s are i.i.d. and integrable random variables such that for any $ 1\le l\le N_C$ and $\forall k\in \mathcal{T}'_l$, we have $X^{(l)}_k \in [-1/d,1/d]$. As such, Hoeffding's inequality yields
\begin{align*}
\mathbb{P}\left(\left |\frac{1}{M}\sum_{k\in \mathcal{T}'_l}\left(X^{(l)}_k-\mathbb{E} \left(X^{(l)}_k\right)\right)\right|>t\right)&\leq 2\, \exp\left(-\frac{M \, t^2}{2/d^2}\right)\\
&=2\,\exp\left(-\frac{1}{2}M \, d^2 \, t^2\right)\\
&\overset{(a)}{=}2\exp(-n^{\epsilon}),
\end{align*}  
where $(a)$ is true if $t=\frac{1}{d}\sqrt{\frac{2n^{\epsilon}}{{M}}}$. Therefore, we have
\begin{align}\label{eq:Exp_2}
\left| \frac{1}{M} \sum_{k\in \mathcal{T}'_l} \left( X^{(l)}_k-\mathbb{E} \left(X^{(l)}_k\right) \right) \right|&\leq \frac{1}{d}\sqrt{\frac{2n^{\epsilon}}{{M}}}
\end{align}
with probability $\ge 1-2\exp(-n^{\epsilon})$. Combining \eqref{eq:Exp_1} and \eqref{eq:Exp_2}, we can upper bound \eqref{eq:Expectation_X1} as follows
\begin{align*}
\left|\sum_{k\in \mathcal{T}'_l} \frac{\cos(2\pi (r_{jk}-x_k))}{r_{jk}}\right|&\leq M \, \left|\frac{1}{M}\sum_{k\in \mathcal{T}'_l} \left(X^{(l)}_k-\mathbb{E} \left(X^{(l)}_k\right)\right) \right|+M \, \left|\mathbb{E} \left(X^{(l)}_1\right)\right|\\
&\leq \frac{M}{d}\sqrt{\frac{2n^{\epsilon}}{{M}}}+\frac{9c_1}{\pi c_2} \frac{M}{l \, d \, n^{\epsilon}}.
\end{align*}
Finally, we have
\begin{align*}
\left|\sum_{\substack{l=1\\ l\neq i}}^{N_C}\sum_{k\in \mathcal{T}_l} \frac{\exp(2\pi i (r_{jk}-x_k))}{r_{jk}}\right|&\leq 2\sum_{l=1}^{N_C}\left|\sum_{k\in \mathcal{T}'_l} \frac{\cos(2\pi (r_{jk}-x_k))}{r_{jk}}\right|+
2\sum_{{l=1}}^{N_C}\left|\sum_{k\in \mathcal{T}'_l} \frac{\sin(2\pi (r_{jk}-x_k))}{r_{jk}}\right|\\
&\overset{(a)}{\leq} 4\sum_{{l=1}}^{N_C}\left(\frac{M}{d}\sqrt{\frac{2n^{\epsilon}}{{M}}}+\frac{9c_1}{\pi c_2} \frac{M}{l \, d \, n^{\epsilon}}\right)\\
&\leq 4\sqrt{2}\,\frac{N_C \, \sqrt{n^{\epsilon} \, M}}{d} + \frac{36c_1}{\pi c_2} \frac{M}{d \, n^{\epsilon}}\log n\\
&\le \left(4\sqrt{2} \frac{N_C\,n^{3\epsilon/2}}{\sqrt{M}\,\log n}+\frac{36c_1}{\pi c_2}\right) \frac{M}{d \, n^{\epsilon}}\log n\\
&=\left(\Theta\left(\frac{n^{1/4-\epsilon}n^{3\epsilon/2}}{n^{3/8}\log n}\right)+\Theta(1)\right) \frac{M}{d \, n^{\epsilon}}\log n = \Theta\left(\frac{M}{d \, n^{\epsilon}}\log n\right),
\end{align*}
where $(a)$ is true with high probability (more precisely, with probability $\ge 1-4\, N_C \, \exp(-n^{\epsilon})$), which concludes the proof.
\end{proof}

\begin{proof}[Proof of Lemma \ref{lem:Gersgorin}]
- Let us first consider the case where $B$ is a Hermitian and positive semi-definite matrix. Then $\Vert B \Vert=\lambda_{\max}(B)$, the largest eigenvalue of $B$. Let now $\lambda$ be an eigenvalue of $B$ and $u$ be its corresponding eigenvector, so that $\lambda u = Bu$. Using the block representation of the matrix $B$, we have
$$
\lambda \, u_j = \sum_{k=1}^K B_{jk} \, u_k, \quad \forall 1 \le j \le K
$$
where $u_j$ is the $j^{th}$ block of the vector $u$. Let now $j$ be such that $\Vert u_j \Vert = \max_{1 \le k \le K} \Vert u_k \Vert$. Taking norms and using the triangle inequality, we obtain
\begin{align*}
|\lambda| \, \Vert u_j \Vert & = \left\Vert \sum_{k=1}^K B_{jk} \, u_k \right\Vert \le \sum_{k=1}^K \Vert B_{jk} \, u_k \Vert\\
& \le \sum_{k=1}^K \Vert B_{jk} \Vert \, \Vert u_k \Vert \le \sum_{k=1}^K \Vert B_{jk} \Vert \, \Vert u_j \Vert
\end{align*}
by the assumption made above. As $u \not\equiv 0$, $\Vert u_j \Vert >0$, so we obtain
$$
|\lambda| \le \max_{1 \le j \le K} \sum_{k=1}^K \Vert B_{jk} \Vert
$$
As this inequality applies to any eigenvalue $\lambda$ of $B$ and $\Vert B \Vert=\lambda_{\max}(B)$, the claim is proved in this case.

- In the general case, observe first that $\Vert B \Vert^2=\lambda_{\max}(BB^{\dagger})$, where $BB^{\dagger}$ is Hermitian and positive semi-definite. So by what was just proved above,
$$
\Vert B \Vert^2 = \lambda_{\max}(BB^{\dagger}) \le \max_{1 \le j \le K} \sum_{k=1}^K \Vert (BB^{\dagger})_{jk} \Vert
$$
Now, $(BB^{\dagger})_{jk} = \sum_{l=1}^K B_{jl} B_{kl}^{\dagger}$ so
\begin{align*}
& \sum_{k=1}^K \Vert (BB^{\dagger})_{jk} \Vert = \sum_{k=1}^K \left\Vert \sum_{l=1}^K B_{jl} B_{kl}^{\dagger} \right\Vert\\
& \le \sum_{k=1}^K \sum_{l=1}^K \Vert B_{jl} \Vert \, \Vert B_{kl} \Vert \le \sum_{l=1}^K \Vert B_{jl} \Vert \, \max_{1 \le j \le K} \sum_{k=1}^K \Vert B_{kj} \Vert 
\end{align*}
and we finally obtain
$$
\Vert B \Vert^2 \le \left( \max_{1 \le j \le K} \sum_{l=1}^K \Vert B_{jl} \Vert  \right) \,  \left( \max_{1 \le j \le K} \sum_{k=1}^K \Vert B_{kj} \Vert \right)
$$
which implies the result, as $ab \le \max\{a,b\}^2$ for any two positive numbers $a,b$.
\end{proof}
\begin{figure}
\centering
\includegraphics[scale=1]{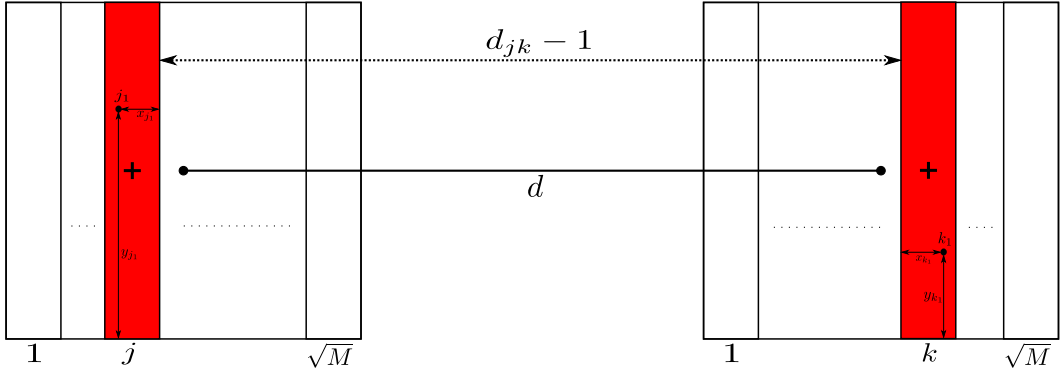}
\caption{Two square clusters that have a center-to-center distance $d$, with each cluster decomposed into  $\sqrt{M}$ vertical $\sqrt{M}\times 1$ rectangles. $d_{jk}$ is distance between the centers (marked with cross) of the two rectangles $j$ and $k$. Moreover, we have the points $j_1(x_{j_1},y_{j_1})$ and $k_1(x_{k_1},y_{k_1})$ in the rectangles $j$ and $k$, respectively.}
\label{fig:4}
\end{figure}
\begin{proof}[Proof of Lemma \ref{lem:blockNorm}]
As in the case of $\Vert H \Vert$, analyzing directly the the asymptotic behavior of $\Vert \widehat{H} \Vert$ reveals itself difficult. We therefore decompose our proof into simpler subproblems. The strategy is essentially the following: in order to bound $\Vert \widehat{H} \Vert$, we divide the matrix into smaller blocks, bound the smaller blocks $\Vert {\widehat{H}}_{jk} \Vert$, and apply Lemma \ref{lem:Gersgorin}. Let us therefore decompose each of the two square clusters into $\sqrt{M}$ vertical $\sqrt{M}\times 1$ rectangles of $\sqrt{M}$ nodes each (See Fig. \ref{fig:4}).

By Lemma \ref{lem:Gersgorin}, we obtain
\begin{equation} \label{eq:ub_H0}
\Vert \widehat{H} \Vert \le \max \left\{ \max_{1 \le j \le \sqrt{M}} \sum_{k=1}^{\sqrt{M}} \Vert \widehat{H}_{jk} \Vert, \max_{1 \le j \le \sqrt{M}} \sum_{k=1}^{\sqrt{M}} \Vert \widehat{H}_{kj} \Vert \right\}
\end{equation}
where the $M \times M$ matrix $\widehat{H}$ is decomposed into blocks $\widehat{H}_{jk}$, $j,k=1,\ldots,\sqrt{M}$, with $\widehat{H}_{jk}$ denoting the $\sqrt{M} \times \sqrt{M}$ channel matrix between $k$-th rectangle of the transmitting cluster and the $j$-th rectangle of the receiving cluster. As shown in Fig. \ref{fig:4}, let us also denote by $d_{jk}$ the corresponding inter-rectangle distance, measured from the centers of the two rectangles. We want to show that for $2\sqrt{M}\le d \le M$, where $d$ is the distance between the centers of the two clusters, there exist constants $c,c'>0$ such that
\begin{equation}\label{eq:up_H0jk}
\Vert \widehat{H}_{jk} \Vert^2 \le c' \, \frac{M^{\epsilon}}{d_{jk}} \le c \, \frac{M^{\epsilon}}{d}
\end{equation}
with high probability as $M \to \infty$. Applying \eqref{eq:ub_H0} and \eqref{eq:up_H0jk}, we get
$$
\Vert \widehat{H} \Vert \le \max \left\{ \max_{1 \le j \le \sqrt{M}} \sum_{k=1}^{\sqrt{M}} \Vert \widehat{H}_{jk} \Vert, \max_{1 \le j \le \sqrt{M}} \sum_{k=1}^{\sqrt{M}} \Vert \widehat{H}_{kj} \Vert \right\} \le \left(c\,\frac{M^{1+\epsilon}}{d}\right)^{1/2}
$$
Therefore, what remains to be proven is inequality \eqref{eq:up_H0jk}. The strategy we propose in order to upper bound $\Vert \widehat{H}_{jk} \Vert^2$ is to use the moments' method, relying on the following inequality:
\begin{align*}
\Vert \widehat{H}_{jk} \Vert^2  & = \lambda_{\max}(\widehat{H}_{jk}\widehat{H}_{jk}^{\dagger}) \le \left( \sum_{k=1}^M (\lambda_k(\widehat{H}_{jk} \widehat{H}_{jk}^{\dagger}))^\ell \right)^{1/\ell}\\
& = \left( \Tr \left((\widehat{H}_{jk} \widehat{H}_{jk}^{\dagger})^\ell \right) \right)^{1/\ell}
\end{align*}
valid for any $\ell \ge 1$. So by Jensen's  inequality, we obtain that $\EE( \Vert \widehat{H}_{jk} \Vert^2) \le \left( \EE( \mathrm{Tr}((\widehat{H}_{jk}\widehat{H}_{jk}^{\dagger})^\ell) ) \right)^{1/\ell}$. In what follows, we show that taking $\ell \to \infty$ leads to $\EE( \Vert \widehat{H}_{jk} \Vert^2) \le c \, \frac{\log M}{d_{jk}}$. More precisely, we show that 
\begin{equation}\label{eq:final}
\EE( \mathrm{Tr}((\widehat{H}_{jk}\widehat{H}_{jk}^{\dagger})^\ell)\le \frac{M(c\,\log M)^{\ell-1}}{d_{jk}^{\ell+1}}
\end{equation}
which implies 
\begin{equation*}
\left( \EE( \mathrm{Tr}((\widehat{H}_{jk}\widehat{H}_{jk}^{\dagger})^\ell) ) \right)^{1/\ell} \le \frac{M^{1/\ell}(c\,\log M)^{1-{1/\ell}}}{d_{jk}^{1+{1/\ell}}} \underset{\ell\rightarrow\infty}{\rightarrow} c\,\frac{\log M}{d_{jk}}. 
\end{equation*}
We first prove \eqref{eq:final} for $\ell=\{1,2\}$, then generalize it to any $\ell$. To simplify the notation, let $F=\widehat{H}_{jk}$. For $\ell=1$, we obtain
\begin{equation}\label{eq:Tr_1}
\EE(\mathrm{Tr}(FF^{\dagger}))=\sum_{j_1,k_1=1}^{\sqrt{M}}\EE(f_{j_1k_1}f_{j_1k_1}^*) =\sum_{j_1,k_1=1}^{\sqrt{M}}\EE(|f_{j_1k_1}|^2) =\sum_{j_1,k_1=1}^{\sqrt{M}}\frac{1}{r_{j_1k_1}^2}\le\frac{M}{d_{jk}^2}
\end{equation}
Note here that given the definition of $d_{jk}$, it only holds that $r_{j_1k_1} \ge d_{jk}-1$ and not $d_{jk}$. However, given our assumption that $d_{jk} \ge \sqrt{M}$, this simplification does not matter asymptotically and also allows to lighten the notation. We will make  this simplification constantly in the following. For $\ell=2$, we obtain 
\begin{align*}
\EE(\mathrm{Tr}((FF^{\dagger})^2))&=\EE(\mathrm{Tr}(FF^{\dagger}FF^{\dagger}))\\
&=\sum_{\substack{j_1,j_2,k_1,k_2=1}}^{\sqrt{M}} \EE(f_{j_1k_1}f_{j_2k_1}^*f_{j_2k_2}f_{j_1k_2}^*) \\
&\le\sum_{\substack{j_1=j_2\\k_1,k_2}} \EE(f_{j_1k_1}f_{j_2k_1}^*f_{j_2k_2}f_{j_1k_2}^*) +\sum_{\substack{j_1, j_2\\k_1=k_2}} \EE(f_{j_1k_1}f_{j_2k_1}^*f_{j_2k_2}f_{j_1k_2}^*) +\sum_{\substack{j_1\neq j_2\\k_1\neq k_2}} \EE(f_{j_1k_1}f_{j_2k_1}^*f_{j_2k_2}f_{j_1k_2}^*)\\
&\leq 2\frac{M^{3/2}}{d_{jk}^4} + M^2 S_2 \overset{(a)}{\le} 2\frac{M}{d_{jk}^3} + M^2 S_2
\end{align*}
where $S_2= |\EE(f_{j_1k_1}f_{j_2k_1}^*f_{j_2k_2}f_{j_1k_2}^*)|$ with $j_1\ne j_2$ and $k_1\neq k_2$ does not depend on the specific choice of $j_1\neq j_2$ and $k_1\ne k_2$, and $(a)$ results from fact that $d_{jk}\ge \sqrt{M}$. In what follows, we upper bound $S_2$. 
\begin{align}\label{eq:S2}\nonumber
S_2 & = |\EE(f_{j_1k_1}f_{j_2k_1}^*f_{j_2k_2}f_{j_1k_2}^*)|\\
& = \bigg| \frac{1}{M^2} \int_0^{1} dx_{j_1}\int_{0}^{\sqrt{M}} dy_{j_1}\int_0^{1} dx_{j_2} \int_{0}^{\sqrt{M}} dy_{j_2} \int_0^{1} dx_{k_1}\int_{0}^{\sqrt{M}} dy_{k_1}\int_0^{1} dx_{k_2} \int_{0}^{\sqrt{M}} dy_{k_2}\, \frac{e^{ 2\pi i (g_{j_1j_2}(k_1)+g_{j_2j_1}(k_2))}}{\rho_{j_1j_2}(k_1)*\rho_{j_2j_1}(k_2)}\bigg|,
\end{align}
where 
\begin{align}\nonumber
g_{j_1j_2}(k_1)&=r_{j_1k_1}-r_{j_2k_1}=-g_{j_2j_1}(k_1)\\
&= \sqrt{(d_{jk}-1+x_{j_1}+x_{k_1})^2+(y_{j_1}-y_{k_1})^2}-\sqrt{(d_{jk}-1+x_{j_2}+x_{k_1})^2+(y_{j_2}-y_{k_1})^2}\label{eq:S2_g}
\end{align}
and
\begin{align}
\rho_{j_1j_2}(k_1) = r_{j_1k_1} \cdot r_{j_2k_1}=\rho_{j_2j_1}(k_1)\ge d_{jk}^2\label{eq:S2_rho},
\end{align}
where $0 \le x_{j_1},x_{j_2},x_{k_1},x_{k_2}\le 1$ and $0 \le y_{j_1},y_{j_2},y_{k_1},y_{k_2}\le \sqrt{M}$ are the horizontal and the vertical positions, respectively (see Fig. \ref{fig:4}).

From now on, let us use the short-hand notation
$$
\int dj \quad \text{for} \quad  \int_{0}^{1} dx_j \int_{0}^{\sqrt{M}} dy_j
$$
Using this short-hand notation as well as equations \eqref{eq:S2_g} and \eqref{eq:S2_rho}, we can rewrite \eqref{eq:S2} as follows 
\begin{align}\nonumber
S_2&= \bigg| \frac{1}{M^2} \int d{j_1} \int d{j_2} \int d{k_1}\, \frac{e^{ 2\pi i g_{j_1j_2}(k_1)}}{\rho_{j_1j_2}(k_1)} \int d{k_2}\, \frac{e^{2\pi i g_{j_2j_1}(k_2)}}{\rho_{j_2j_1}(k_2)}\bigg|\\\nonumber
&\leq \frac{1}{M^2}\int d{j_1} \int d{j_2} \,\bigg| \int d{k_1}\, \frac{e^{2\pi i g_{j_1j_2}(k_1)}}{\rho_{j_1j_2}(k_1)}\bigg| \cdot \bigg|\int d{k_2}\, \frac{e^{2\pi i g_{j_2j_1}(k_2)}}{\rho_{j_2j_1}(k_2)}\bigg|\\\nonumber
&= \frac{1}{M^2}\int d{j_1} \int d{j_2} \,\bigg| \int d{k_1}\, \frac{e^{2\pi i g_{j_1j_2}(k_1)}}{\rho_{j_1j_2}(k_1)}\bigg| \cdot B_{2,1}\\\nonumber
\end{align}
where
\begin{align}\label{eq:B21}\nonumber
B_{2,1}=\bigg|\int d{k_2}\, \frac{e^{2\pi i g_{j_2j_1}(k_2)}}{\rho_{j_2j_1}(k_2)}\bigg|&=\bigg|\int_{0}^{1} dx_{k_2}\int_{0}^{\sqrt{M}} dy_{k_2}\, \frac{e^{2\pi i g_{j_2j_1}(k_2)}}{\rho_{j_2j_1}(k_2)}\bigg|\\\nonumber
&\le \int_{0}^{1} dx_{k_2}\int_{0}^{\sqrt{M}} dy_{k_2}\, \bigg|\frac{e^{2\pi i g_{j_2j_1}(k_2)}}{\rho_{j_2j_1}(k_2)}\bigg| \nonumber\\
&= \int_{0}^{1} dx_{k_2}\int_{0}^{\sqrt{M}} dy_{k_2}\, \frac{1}{\rho_{j_2j_1}(k_2)}\leq \frac{\sqrt{M}}{d_{jk}^2}
\end{align} 
We therefore obtain
\begin{equation}
S_2 \le \frac{1}{M^{3/2} d_{jk}^2} \int d{j_1} \cdot A_{1,2} \label{eq:S2_UB}
\end{equation}
where
$$
A_{1,2} =  \int d{j_2} \,\bigg| \int d{k_1}\, \frac{e^{2\pi i g_{j_1j_2}(k_1)}}{\rho_{j_1j_2}(k_1)}\bigg| 
$$
Before further upper bounding \eqref{eq:S2_UB}, we present the following lemma, taken from \cite{JSAC} and adapted to the present situation.

\begin{lem} \label{lem:mu}
Let $g:[0,\sqrt{M}] \to \mathbb{R}$ be a $C^2$ function such that $|g'(y)| \geq c_1>0$ for all $z \in [0,\sqrt{M}]$ and $g''$ changes sign at most twice on $[0,\sqrt{M}]$ (say e.g.~$g''(y) \ge 0 $ in $[y_-,y_+]$ and $g''(y) \le 0$ outside). Let also $\rho:[0,\sqrt{M}] \to \mathbb{R}$ be a $C^1$ function such that $|\rho(y)| \ge c_2>0$ and $\rho'(y)$ changes sign at most twice on $[0,\sqrt{M}]$. Then
$$
\left| \int_0^{\sqrt{M}} dy \, \frac{e^{ 2 \pi i g(y)}}{\rho(y)} \right| \le \frac{7}{\pi \, c_1 \, c_2}.
$$
\end{lem}

\begin{proof}
By the integration by parts formula, we obtain
\begin{align*}
\int_0^{\sqrt{M}} dy \, \frac{e^{2 \pi i g(y)}}{\rho(y)} &=
\int_0^{\sqrt{M}} dy \, \frac{2 \pi i g'(y)}{2 \pi i  g'(y)\rho(y)} \, e^{ 2 \pi i  g(y)}\\
&=  \frac{e^{ 2 \pi i g(y))}}{2 \pi i g'(y)\rho(y)} \bigg|_0^{\sqrt{M}} - \int_0^{\sqrt{M}} dy \, \frac{g''(y)\rho(y)+g'(y)\rho'(y)}{2 \pi i (g'(y)\rho(y))^2} \, e^{ 2 \pi i g(y)}
\end{align*}
which in turn yields the upper bound
\begin{align*}
\left|\int_0^{\sqrt{M}} dy \, \frac{e^{2 \pi i g(y)}}{\rho(y)} \right|&\le \frac{1}{2 \pi} \, \Bigg( \frac{1}{|g'({\sqrt{M}})||\rho({\sqrt{M}})|}+ \frac{1}{|g'(0)||\rho(0)|}\\&+ \int_0^{\sqrt{M}} dy \, \frac{|g''(y)|}{(g'(y))^2|\rho(y)|} + \int_0^{\sqrt{M}} dy \, \frac{|\rho'(y)|}{g'(y)(\rho(y))^2} \Bigg)
\end{align*}
By the assumptions made in the lemma, we have
\begin{align*}
\int_0^{\sqrt{M}} dy \, \frac{|g''(y)|}{(g'(y))^2|\rho(z)|} &\le \frac{1}{c_2} \int_0^{\sqrt{M}} dy \, \frac{|g''(y)|}{(g'(y))^2}\\
&= \frac{1}{c_2} \Bigg( - \int_0^{y_-} dy \, \frac{g''(y)}{(g'(y))^2} + \int_{y_-}^{y_+} dy \, \frac{g''(y)}{(g'(y))^2}
 - \int_{y_+}^{\sqrt{M}} dy \, \frac{g''(y)}{(g'(y))^2} \Bigg)\\
&= \frac{1}{c_2} \, \Bigg(
\frac{1}{g'({\sqrt{M}})} - \frac{1}{g'(0)} + \frac{2}{g'(y_-)} - \frac{2}{g'(y_+)} \Bigg)
\end{align*}
So
$$
\int_0^{\sqrt{M}} dy \, \frac{|g''(y)|}{(g'(y))^2|\rho(y)|} \le \frac{7}{c_1 \, c_2}.
$$
We obtain in a similar manner that
$$
\int_0^{\sqrt{M}} dy \, \frac{|\rho'(y)|}{g'(y)(\rho(y))^2} \le \frac{7}{c_1 \, c_2}
$$
Combining all the bounds, we finally get
$$
\left| \int_0^{\sqrt{M}} dy \, \frac{e^{2 \pi i g(y)}}{\rho(y)} \right| \le \frac{7}{\pi \, c_1 \, c_2}
$$
\end{proof}
For any $\epsilon>0$, we can upper bound $A_{1,2}$ in equation \eqref{eq:S2_UB} as follows
\begin{align}\nonumber
A_{1,2}&= \int dj_2\bigg|\int d{k_1}\, \frac{e^{2\pi i g_{j_1j_2}(k_1)}}{\rho_{j_1j_2}(k_1)}\bigg|\\\nonumber
&= \int_{|y_{j_2}-y_{j_1}|<\epsilon\sqrt{M}} d{j_2} \bigg| \int d{k_1}\, \frac{e^{2\pi i g_{j_1j_2}(k_1)}}{\rho_{j_1j_2}(k_1)}\bigg| + \int_{|y_{j_2}-y_{j_1}|\geq\epsilon\sqrt{M}} d{j_2} \bigg| \int  d{k_1}\, \frac{e^{2\pi i g_{j_1j_2}(k_1)}}{\rho_{j_1j_2}(k_1)}\bigg|\\\nonumber
&\le \int_{|y_{j_2}-y_{j_1}|<\epsilon\sqrt{M}} d{j_2}  \int d{k_1}\, \frac{1}{\rho_{j_1j_2}(k_1)} + \int_{|y_{j_2}-y_{j_1}|\geq\epsilon\sqrt{M}} d{j_2} \bigg| \int  d{k_1}\, \frac{e^{2\pi i g_{j_1j_2}(k_1)}}{\rho_{j_1j_2}(k_1)}\bigg|\\
&\le \frac{\epsilon\,M}{d_{jk}^2} + \int_{|y_{j_2}-y_{j_1}|\geq\epsilon\sqrt{M}} d{j_2} \bigg| \int  d{k_1}\, \frac{e^{2\pi i g_{j_1j_2}(k_1)}}{\rho_{j_1j_2}(k_1)}\bigg|\label{eq:A12}
\end{align}
Furthermore, note that 
\begin{align*}
g_{j_1j_2}(k_1)&=r_{j_1k_1}-r_{j_2k_1}\\
&=-\int_{x_{j_1}}^{x_{j_2}}\frac{d_{jk}-1+x+x_{k_1}}{\sqrt{(d_{jk}-1+x+x_{k_1})^2+(y_{j_1}-y_{k_1})^2}}dx+
\int_{y_{j_1}}^{y_{j_2}}\frac{y_{k_1}-y}{\sqrt{(d_{jk}-1+x_{j_2}+x_{k_1})^2+(y-y_{k_1})^2}}dy
\end{align*}
Therefore, the first order partial derivative of $g_{j_1j_2}(k_1)$ with respect to $y_{k_1}$ is given by
\begin{align*}
\frac{\partial g_{j_1j_2}(k_1)}{\partial y_{k_1}} =\int_{x_{j_1}}^{x_{j_2}}\frac{(y_{k_1}-y_{j_1})(d_{jk}-1+x+x_{k_1})}{\left((d_{jk}-1+x+x_{k_1})^2+(y_{j_1}-y_{k_1})^2\right)^{3/2}}dx+
\int_{y_{j_1}}^{y_{j_2}}\frac{(d_{jk}-1+x_{j_2}+x_{k_1})^2}{\left((d_{jk}-1+x_{j_2}+x_{k_1})^2+(y-y_{k_1})^2\right)^{3/2}}dy
\end{align*}
From this expression, we deduce that for a constant $c_3>0$
\begin{align}\nonumber\label{eq:dg_LB}
\bigg|\frac{\partial g_{j_1j_2}(k_1)}{\partial y_{k_1}}\bigg| &\geq c_3\,\frac{|y_{j_2}-y_{j_1}|}{d_{jk}}-\frac{|y_{k_1}-y_{j_1}|.|x_{j_2}-x_{j_1}|}{d_{jk}^2} \\\nonumber
&\ge c_3\,\frac{|y_{j_2}-y_{j_1}|}{d_{jk}}-\frac{\sqrt{M}}{d_{jk}^2}\\
&\overset{(a)}{\ge} \frac{c_3\,|y_{j_2}-y_{j_1}|-1}{d_{jk}},
\end{align}
where $(a)$ follows from the fact that $d_{jk}\geq \sqrt{M}$. For $c_3\,|y_{j_2}-y_{j_1}|-1>0$ (we will tune $\epsilon$ accordingly, as we will see), using \eqref{eq:S2_rho} and \eqref{eq:dg_LB}, we can apply lemma \ref{lem:mu} and upper bound the second term in \eqref{eq:A12} as follows
\begin{align}\nonumber
\int_{|y_{j_2}-y_{j_1}|\geq\epsilon\sqrt{M}} d{j_2} \bigg| \int d{k_1}\, \frac{e^{2\pi i g_{j_1j_2}(k_1)}}{\rho_{j_1j_2}(k_1)}\bigg|&\le \int_{|y_{j_2}-y_{j_1}|\geq\epsilon\sqrt{M}} d{j_2}\int_0^1 dx_{k_1} \bigg| \int_0^{\sqrt{M}} dy_{k_1}\, \frac{e^{2\pi i g_{j_1j_2}(k_1)}}{\rho_{j_1j_2}(k_1)}\bigg|\\\nonumber
&\le  \int_{|y_{j_2}-y_{j_1}|\geq\epsilon\sqrt{M}} dy_{j_2} \frac{7}{\pi \,\frac{c_3|y_{j_2}-y_{j_1}|-1}{d_{jk}} d_{jk}^2}\\\nonumber
&\le \frac{7}{\pi c_3 d_{jk}} \int_{|y_{j_2}-y_{j_1}|\geq\epsilon\sqrt{M}} \frac{1}{|y_{j_2}-y_{j_1}|-1/c_3}dy_{j_2}\\\label{eq:UB_to_use}
&\le \frac{7}{\pi c_3d_{jk}}\log\left(\frac{1}{\epsilon}\right)
\end{align}
which gives the following upper bound on \eqref{eq:A12}
\begin{equation}\label{eq:A12_UB}
A_{1,2}\le \frac{\epsilon\,M}{d_{jk}^2} + \frac{7}{\pi c_3d_{jk}}\log\left(\frac{1}{\epsilon}\right)\overset{(a)}{=} O\left(\frac{\log M}{d_{jk}}\right),
\end{equation}
where $(a)$ results from choosing $\epsilon=\frac{c_4}{\sqrt{M}}$ with sufficiently large $c_4>0$, which also ensures that $c_3\,|y_{j_2}-y_{j_1}|-1>0$. For the chosen value of $\epsilon$, we get $S_2= O\left(\frac{1}{\sqrt{M} d_{jk}^4}\right)+ O\left(\frac{1}{M d_{jk}^3}\log M\right)=O\left(\frac{1}{M d_{jk}^3}\log M\right)$. As a result, we get 
\begin{equation}\label{eq:Tr_2}
\EE(\mathrm{Tr}((FF^{\dagger})^2))\leq 2\frac{M}{d_{jk}^3} + M^2 S_2 =O\left(M \, \frac{c\,\log M}{d_{jk}^3}\right).
\end{equation} 
Now, we generalize our result to any moment $\ell>2$. We start with the following 
proposition.
\begin{lem}\label{lem:l_th-moment}
\begin{align*}
\EE(\mathrm{Tr}((FF^{\dagger})^{\ell}))&\leq \frac{2\,\ell}{\sqrt{M}} \sum_{t=1}^{\lfloor \ell/2\rfloor}
\EE(\mathrm{Tr}((FF^{\dagger})^{t}))\,\EE(\mathrm{Tr}((FF^{\dagger})^{\ell-t})) + M^{{\ell}} S_{\ell}
\end{align*}
where
\begin{equation}\label{eq:Sl}
S_{\ell} =|\EE(f_{j_1k_1}f_{j_2k_1}^*\ldots f_{j_{\ell}k_{\ell}}f_{j_1k_{\ell}}^*)|,
\end{equation}
with $j_1\neq \ldots \neq j_{\ell}$ and $k_1\neq \ldots \neq k_{\ell}$. Note that $S_l$ does not depend on the particular choice of  $j_1\neq \ldots \neq j_{\ell}$ and $k_1\neq \ldots \neq k_{\ell}$.
\end{lem}

\begin{proof}
We know that
$$\EE(\mathrm{Tr}((FF^{\dagger})^\ell))=\sum_{\substack{j_1, \ldots, j_{\ell}=1\\k_1, \ldots, k_{\ell}=1}}^{\sqrt{M}} \EE(f_{j_1k_1}f_{j_2k_1}^*\ldots f_{j_{\ell}k_{\ell}}f_{j_1k_{\ell}}^*).$$
We split the summation such that at least two $j_t$ indices are equal or two $k_t$ indices are equal, where $t\in\{1,\ldots,\ell\}$. As such, we can have
\begin{align*}
\EE(\mathrm{Tr}((FF^{\dagger})^\ell))&\le \sum_{\substack{j_1=j_2=1,\\j_3, \ldots, j_{\ell}=1,\\k_1, \ldots, k_{\ell}=1}}^{\sqrt{M}} \EE(f_{j_1k_1}f_{j_2k_1}^*\ldots f_{j_{\ell}k_{\ell}}f_{j_1k_{\ell}}^*)+\sum_{\substack{j_{1}=j_{3}=1,\\j_2,j_4, \ldots, j_{\ell}=1,\\k_1, \ldots,k_{\ell}=1}}^{\sqrt{M}} \EE(f_{j_1k_1}f_{j_2k_1}^*\ldots f_{j_{\ell}k_{\ell}}f_{j_1k_{\ell}}^*)+\ldots\\
&+\sum_{\substack{j_1\neq \ldots\neq j_{\ell}\\k_1\neq \ldots\neq k_{\ell}}} \EE(f_{j_1k_1}f_{j_2k_1}^*\ldots f_{j_{\ell}k_{\ell}}f_{j_1k_{\ell}}^*)\\
&\overset{(a)}\leq 2\,\ell \sum_{t=1}^{\lfloor \ell/2 \rfloor} \sum_{\substack{j_1,\ldots,j_t,\\j_{t+1}=j_1,\\j_{t+2}\ldots, j_{\ell}=1,\\k_1, \ldots, k_{\ell}=1}}^{\sqrt{M}} \EE(f_{j_1k_1}\ldots f_{j_1k_t}^*f_{j_1k_{t+1}}\ldots f_{j_1k_{\ell}}^*) + \sum_{\substack{j_1\neq \ldots\neq j_{\ell}\\k_1\neq \ldots\neq k_{\ell}}} \EE(f_{j_1k_1}f_{j_2k_1}^*\ldots f_{j_{\ell}k_{\ell}}f_{j_1k_{\ell}}^*)\\
&\overset{(b)}= 2\,\ell\sum_{t=1}^{\lfloor \ell/2\rfloor}\sum_{j_1=1}^{\sqrt{M}}
\frac{\EE(\mathrm{Tr}((FF^{\dagger})^{t}))}{\sqrt{M}}\, \frac{\EE(\mathrm{Tr}((FF^{\dagger})^{\ell-t}))}{\sqrt{M}} + \sum_{\substack{j_1\neq \ldots\neq j_{\ell}\\k_1\neq \ldots\neq k_{\ell}}} \EE(f_{j_1k_1}f_{j_2k_1}^*\ldots f_{j_{\ell}k_{\ell}}f_{j_1k_{\ell}}^*)\\
&=\frac{2\,\ell}{\sqrt{M}}\sum_{t=1}^{\lfloor \ell/2\rfloor}
\EE(\mathrm{Tr}((FF^{\dagger})^{t}))\, \EE(\mathrm{Tr}((FF^{\dagger})^{\ell-t})) + M^{{\ell}} S_{\ell}
\end{align*}
where 
$S_{\ell}$ is defined in \eqref{eq:Sl} and $(a)$ follows from the fact that we have at most $\ell$ different ways to choose two $j$ indices (equivalently two $k$ indices) apart by $t$ (for example, the indices $j_{t_1}$ and $j_{t_2}$, where $\min\{|t_2-t_1|,\ell-|t_2-t_1|\}=t$) to be equal. The particular choice of the indices is irrelevant for the computation of the expectation. Moreover, $(b)$ represents an order equality and it is not straight forward, however, we omit the technical details for the sake of the readability of the proof. 
\end{proof}

We assume now that
$$
\EE( \mathrm{Tr}((\widehat{H}_{jk}\widehat{H}_{jk}^{\dagger})^{\ell-1})\le \frac{M(c\,\log M)^{\ell-2}}{d_{jk}^{\ell}},
$$
which holds for the first and the second moments (see equations \eqref{eq:Tr_1} and \eqref{eq:Tr_2}), and prove that it also holds for the ${\ell}$-th moment.
\\

We proceed by upper bounding $S_{\ell}$ given in \eqref{eq:Sl}:
\begin{align*}
S_{\ell}&= \bigg| \frac{1}{M^{\ell}} \int d{j_1} \left(\int d{j_2} \int d{k_1}\, \frac{e^{ 2\pi i g_{j_1j_2}(k_1)}}{\rho_{j_1j_2}(k_1)}\right) 
\left(\int d{j_3} \int d{k_2}\, \frac{e^{ 2\pi i g_{j_2j_3}(k_2)}}{\rho_{j_2j_3}(k_2)}\right)\ldots\\
&\hspace{1cm} \left(\int d{j_{\ell}} \int d{k_{\ell-1}}\, \frac{e^{ 2\pi i g_{j_{\ell-1},j_{\ell}}(k_{\ell-1})}}{\rho_{j_{\ell-1},j_{\ell}}(k_{\ell-1})}\right)
\left(\int d{k_{\ell}}\, \frac{e^{ 2\pi i g_{j_{\ell},j_1}(k_{\ell})}}{\rho_{j_{\ell},j_{1}}(k_{\ell})}\right)
\bigg|\\
&\leq  \frac{1}{M^{\ell}} \int d{j_1} \left(\int d{j_2} \bigg|\int d{k_1}\, \frac{e^{ 2\pi i g_{j_1j_2}(k_1)}}{\rho_{j_1j_2}(k_1)}\bigg|\right) 
\left(\int d{j_3}\bigg| \int d{k_2}\, \frac{e^{ 2\pi i g_{j_2j_3}(k_2)}}{\rho_{j_2j_3}(k_2)}\bigg|\right)\ldots\\
&\hspace{1cm} \left(\int d{j_{\ell}} \bigg|\int d{k_{\ell-1}}\, \frac{e^{ 2\pi i g_{j_{\ell-1},j_{\ell}}(k_{\ell-1})}}{\rho_{j_{\ell-1},j_{\ell}}(k_{\ell-1})}\bigg|\right)
\bigg|\int d{k_{\ell}}\, \frac{e^{ 2\pi i g_{j_{\ell},j_1}(k_{\ell})}}{\rho_{j_{\ell},j_{1}}(k_{\ell})}
\bigg|\\
&=\frac{1}{M^{\ell}}\int d{j_1}\,A_{1,2} \cdot A_{2,3} \cdots  A_{\ell-1,\ell} \cdot B_{\ell,1} 
\end{align*}
where (just as we defined $A_{1,2}$ and $B_{2,1}$) 
$$
A_{t,t+1}=\int d{j_{t}} \bigg|\int d{k_{t-1}}\, \frac{e^{ 2\pi i g_{j_{t-1},j_{t}}(k_{t-1})}}{\rho_{j_{t-1},j_{t}}(k_{t-1})}\bigg|\,\,\,\,\,\, \text{for $1\le t\le \ell-1$}
$$
and
$$
B_{\ell,1}=\bigg|\int d{k_{\ell}}\, \frac{e^{ 2\pi i g_{j_{\ell},j_1}(k_{\ell})}}{\rho_{j_{\ell},j_{1}}(k_{\ell})}
\bigg|.
$$
Similarly  to how we proceeded with $A_{1,2}$ and $B_{2,1}$ in \eqref{eq:A12_UB} and \eqref{eq:B21}, respectively, we now upper bound $A_{t,t+1}$ (for $1\leq t \leq \ell-1$) and $B_{\ell,1}$. Therefore, we get
\begin{align*}
S_{\ell}\leq \frac{1}{M^{\ell-1/2}d_{jk}^2}\int d{j_1}\,A_{1,2} \cdot A_{2,3} \cdots A_{\ell-1,\ell} \cdot B_{\ell,1} \le \frac{1}{M^{\ell-1} d_{jk}^2}\left(c\,\frac{\log{M}}{d_{jk}}\right)^{\ell-1} 
\end{align*}
Finally, we obtain
\begin{align*}
\EE(\mathrm{Tr}((FF^{\dagger})^\ell))&\leq \frac{2\,\ell}{\sqrt{M}}\sum_{t=1}^{\lfloor \ell/2\rfloor}
\EE(\mathrm{Tr}((FF^{\dagger})^{t}))\, \EE(\mathrm{Tr}((FF^{\dagger})^{\ell-t})) + M^{{\ell}} S_{\ell}\\
&\le \frac{2\,\ell}{\sqrt{M}}\sum_{t=1}^{\lfloor \ell/2\rfloor} \frac{M(c\,\log M)^{t-1}}{d_{jk}^{t+1}}\frac{M(c\,\log M)^{\ell-t-1}}{d_{jk}^{\ell-t+1}}+ M^{{\ell}} S_{\ell}\\
&\le O\left( M\frac{(c\,\log M)^{{\ell}-1}}{d_{jk}^{{\ell}+1}}\right) + M^{{\ell}} S_{\ell} = O\left( M\frac{(c\,\log M)^{{\ell}-1}}{d_{jk}^{{\ell}+1}}\right),
\end{align*}
which concludes the induction. The last step includes applying Markov's inequality to get
\begin{align*}
\mathbb{P}\left(\lambda_{\max}(\widehat{H}_{jk}\widehat{H}_{jk}^{\dagger})\geq 
c' \frac{M^{\epsilon}}{d_{jk}}\right)&\leq \frac{\EE((\lambda_{\max}(\widehat{H}_{jk}\widehat{H}_{jk}^{\dagger}))^\ell)}
{(c' M^{\epsilon}/d_{jk})^{\ell}}\\
&\le \frac{\EE(\mathrm{Tr}((FF^{\dagger})^\ell))}
{(c' M^{\epsilon}/d_{jk})^{\ell}}\\
&\le \frac{{M\,(c\log M)^{{\ell}-1}}/{d_{jk}^{{\ell}+1}}}{(c' M^{\epsilon}/d_{jk})^{\ell}}\\
&\le \frac{M\,(\log M)^{{\ell}-1}}{d_{jk}\, M^{\epsilon\ell}}
\end{align*}
which, for any fixed $\epsilon > 0$, can be made arbitrarily small by taking $\ell$ sufficiently large.\\
\begin{figure}
\centering
\includegraphics[scale=0.55]{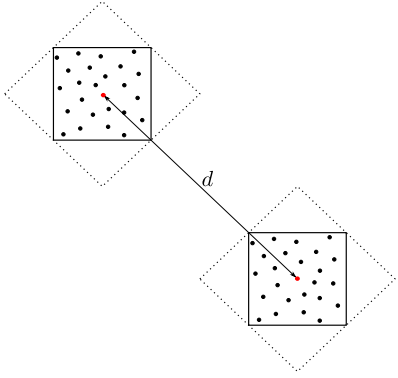}
\caption{Two tilted square clusters that have a center-to-center distance $d$. We can draw larger squares (drawn in dotted line) containing the original clusters with the same centers that are aligned.}
\label{fig:5}
\end{figure}

A last remark is that we proved lemma \ref{lem:blockNorm} for aligned clusters. However, the proof can be easily generalized to tilted clusters, as shown in Fig. \ref{fig:5}. We can always draw a larger cluster containing the original cluster and having the same center. The larger cluster can at most contain twice as many nodes as the original cluster. The large clusters are now aligned. Moreover, the distance $d$ from the centers of the two newly created large clusters still satisfies the required condition ($2\sqrt{M}\le d \le M$).
\end{proof} 

%\bibliographystyle{IEEEtran}
%\bibliography{BIBfile}

\end{document}